\documentclass[conference]{IEEEtran}
\IEEEoverridecommandlockouts
\usepackage{cite}
\usepackage{amsmath,amssymb,amsfonts,amsthm}
\usepackage{algorithmic}
\usepackage{graphicx}
\usepackage{textcomp}
\usepackage{enumerate}
\usepackage{latexsym}
\usepackage{mathtools}
\usepackage{bm}
\usepackage{ascmac}
\def\BibTeX{{\rm B\kern-.05em{\sc i\kern-.025em b}\kern-.08em
    T\kern-.1667em\lower.7ex\hbox{E}\kern-.125emX}}
\DeclarePairedDelimiter{\abs}{\lvert}{\rvert}
\theoremstyle{definition}
\newtheorem{theorem}{Theorem}
\newtheorem{definition}{Definition}
\newtheorem{assumption}{Assumption}

\newtheorem{remark}{Remark}
\newtheorem{example}{Example}

\newcommand{\argmax}{\mathop{\rm arg\,max\,}\limits}

\newcommand{\defeq}{\overset{\mbox{\tiny{\textit{def}}}}{=}}
\newcommand{\ie}{\textrm{i}.\textrm{e}.}
\newcommand{\ipid}{\textrm{i}.\textrm{p}.\textrm{i}.\textrm{d}.}
\newcommand{\iid}{\textrm{i}.\textrm{i}.\textrm{d}.}
\newcommand{\etal}{\textit{et}\,\textit{al}.}

\newcommand{\btheta}{\bm{\theta}}
\title{An Efficient Bayes Coding Algorithm for the Non-Stationary Source in Which Context Tree Model Varies from Interval to Interval
}

\author{\IEEEauthorblockN{Koshi Shimada}
\IEEEauthorblockA{\textit{Department of Pure and Applied Mathematics} \\
\textit{Waseda University}\\
Tokyo, Japan \\
shimada.koshi.re@gmail.com}
\and
\IEEEauthorblockN{Shota Saito}
\IEEEauthorblockA{\textit{Faculty of Informatics} \\
\textit{Gunma University}\\
Gunma, Japan \\
shota.s@gunma-u.ac.jp}
\and
\IEEEauthorblockN{Toshiyasu Matsushima}
\IEEEauthorblockA{\textit{Department of Pure and Applied Mathematics} \\
\textit{Waseda University}\\
Tokyo, Japan \\
toshimat@waseda.jp}
}

\begin{document}

\maketitle

\begin{abstract}
  The context tree source is a source model in which the occurrence probability of symbols is determined from a finite past sequence, and is a broader class of sources that includes {\iid} and Markov sources.
  The proposed source model in this paper represents that a subsequence in each interval is generated from a different context tree model.
  The Bayes code for such sources requires weighting of the posterior probability distributions for the change patterns of the context tree source and all possible context tree models.
  Therefore, the challenge is how to reduce this exponential order computational complexity.
  In this paper, we assume a special class of prior probability distribution of change patterns and context tree models, and propose an efficient Bayes coding algorithm whose computational complexity is the polynomial order.
\end{abstract}


\section{Introduction}

The arithmetic codes asymptotically achieve the minimal expected length for lossless source coding.
The problem with this method is that it cannot be used unless the probabilistic structure of the source is known in advance.
Therefore, universal codes, which can be used when the probability distribution of the source is unknown, have been studied.

The context tree source is one of the major source models for universal coding, and the CTW (Context Tree Weighting)\cite{ctw} is known as an efficient universal code for context tree sources.
The CTW method can be interpreted as a special case of the Bayes code proposed by Matsushima and Hirasawa\cite{matsu1_ctm}.
The CTW method encodes the entire source sequence at once, which is not an efficient use of memory and causes underflow problem in calculation, whereas the Bayes code of Matsushima and Hirasawa\cite{matsu1_ctm} can be encoded sequentially and is free from these problems.
It is known that the Bayes code has equal codeword length when encoded sequentially and when the entire sequence is encoded at once\cite{matsu_bayes}.

However, in the Bayes code, the computational complexity of weighting by the posterior probability of the context tree models increases exponentially according to the maximum depth of the context tree models.
Matsushima and Hirasawa\cite{matsu1_ctm} developed the efficient Bayes coding algorithm by assuming an appropriate class of prior probability distributions for the context tree models, which reduces the computational complexity from the exponential order to the polynomial order.

Now, a source model for source coding should be able to describe in a concise mathematical manner, but also should be better reflects the probability structure of the real data sequence to be compressed.
For example, the context tree source includes {\iid} source and Markov source inherent in itself.
It is a broader class of sources, and has been applied to text data, for example.

On the other hand, there are cases where it is appropriate to think of symbols as being generated according to a different context tree source for each interval, rather than modeling the entire data series as being generated according to a single context tree source.
For example, in the case of the human genome, the DNA sequence consists of about $3$ billion base pairs, and it is described as a pair of series of about $30$ billion in length with four different alphabets: A, G, T, and C.
Although Markov source is sometimes assumed in DNA sequence compression algorithms\cite{cao}, it is known that there are genetic and non-genetic regions in the human genome, which have different structural characteristics.
Therefore, in this paper, we present a non-stationary source that context tree source changes from interval to interval.

An example of a non-stationary source where the source changes from interval to interval is an {\ipid} (independently piecewise identically distributed) source\cite{ipid1,ipid2}.
An {\ipid} source is a source consisting of an {\iid} sequences of parameters that are different for each interval.
It can be regarded as a special case of the proposed source in this paper.
An efficient Bayes code for {\ipid} sources has already been proposed by Suso {\etal}\cite{suko_ipid}.

Assuming the source model that symbols are generated by different context tree sources in each interval, we present an efficient Bayes coding algorithm for it.
In this algorithm, we use the prior probability of context tree models by Matsushima and Hirasawa\cite{matsu1_ctm} and that of parameter change patterns by Suko {\etal}\cite{suko_ipid}.
The proposed algorithm achieves a reduction in computational complexity from the exponential order to the polynomial order.

\section{Non-stationary source that context tree model changes from interval to interval}

In this section, we present a non-stationary source that context tree model changes from interval to interval.
The symbols are generated from different context tree models depending on the interval, as shown in Figure \ref{model_overview}.
\begin{figure}[ht]
  \centering
  \includegraphics[width=0.9\columnwidth, bb=0 0 641 276]{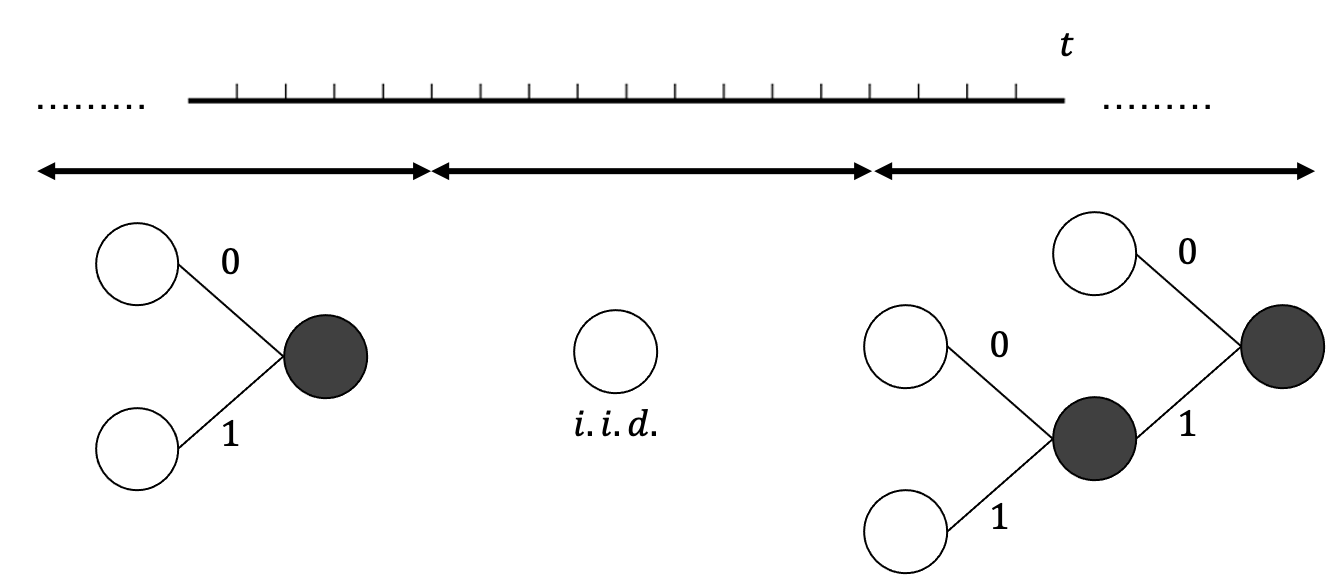}
  \caption{Diagram of a non-stationary source with a context tree model that changes from interval to interval}
  \label{model_overview}
\end{figure}

Now we define the change pattern of context tree models as follows.
\begin{definition}
  The change pattern $c$ is defined to indicate when the context tree model has changed. That is,
  \begin{align} \label{pattern2}
    & c\defeq\left( w_1^{(c)}, \ldots, w_t^{(c)}, \ldots, w_N^{(c)} \right)
    \in\mathcal{C}_N \defeq {\{0,1\}^N}, \\ \nonumber 
    & w_t^{(c)} \defeq
    \begin{cases}
      1 & \mbox{if the context tree model changes at time $t$}, \\
      0 & \mbox{otherwise}.
    \end{cases}
  \end{align}
\end{definition}
Now, for convenience, let $w_1^{(C)}=1$.
The length $N$ of a source sequence is fixed.
The set of all change patterns $\mathcal{C}_N$ is abbreviated as $\mathcal{C}$ from now on.
Next, we define the set of points in the context tree model where changes occur.
\begin{definition}
  Let $\mathcal{T}_c$ denote the set of points at which the parameter changes in the change pattern $c$. That is,
  \begin{gather} \label{time_set2}
    \mathcal{T}_c\defeq\left\{t\,\middle|\,w_t^{(c)}=1\right\} =
    \left\{ t_0^{(c)}, t_1^{(c)}, \ldots,
    t_{\abs{\mathcal{T}_c}-1}^{(c)} \right\}.
  \end{gather}
  where $t_j^{(c)}$ is the $j$-th changing point in the change pattern $c$.
\end{definition}
In other words, there are $\abs{\mathcal{T}_c}-1$ parameter changes in $c$.
For convenience, let $t_0^{(c)}=1$, $t_{\abs{\mathcal{T}_c}}^{(c)}=N+1$.
If the change pattern $c$ is specified in advance, $t_j^{(c)}$ is abbreviated as $t_j$.

From the $j$-th changing point $t_j$ to the $j+1$-th changing point $t_{j+1}$, symbols are generated according to a single context tree model $m_{t_j}^{(c)}$.
The parameter $\btheta^{m_{t_j}^{(c)}}$ for this $m_{t_j}^{(c)}$ is defined as in Definition \ref{parameter_def} and an example is shown in Figure \ref{parameter_setting}.
\begin{definition} \label{parameter_def}
  In the change pattern $c$,
  for the context tree model $m_{t_j}^{(c)}$ in the interval $\left[t_j,t_{j+1}\right)$,
  we denote the set of its leaf nodes as $L_{m_{t_j}}^{(c)}$.
  The parameter $\btheta^{m_{t_j}^{(c)}}$ for $m_{t_j}^{(c)}$ is defined as follows:
  \begin{align}
    \btheta^{m_{t_j}^{(c)}} \defeq \left\{
    \btheta_s \in (0,1)^{\abs{\mathcal{X}}} \,\middle|\,
    s\in L_{m_{t_j}}^{(c)} \right\},
  \end{align}
  where
  \begin{align}
    & \ \btheta_s \defeq {}^\mathsf{T} \hspace{-1mm}
    \left( {\theta_{0|s}}, {\theta_{1|s}}, \ldots,
    \theta_{\abs{\mathcal{X}}-1|s} \right), \\ &
    \sum_{a\in\mathcal{X}} \theta_{a|s} = 1,~~
    \theta_{a|s}\in (0,1)~~\mbox{for each symbol $a$.}
  \end{align}
  Note that $\theta_{a|s}$ is the occurrence probability of $a\in\mathcal{X}$ under the state corresponding to node $s$, where $\mathcal{X}$ denotes a source alphabet.
\end{definition}
\begin{figure}[ht]
  \centering
  \includegraphics[width=0.75\columnwidth, bb=0 0 597 315]{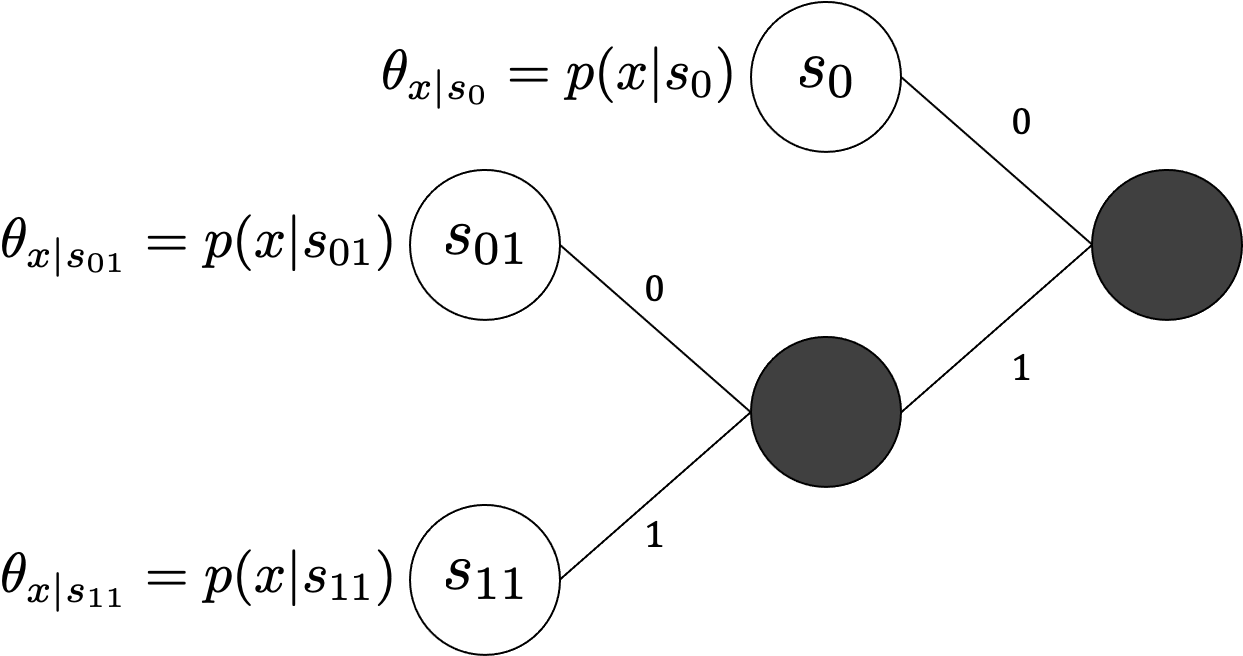}
  \caption{Example of occurrence probability based on context tree model}
  \label{parameter_setting}
\end{figure}
In the case where the change pattern $c$ is specified in advance, $m_{t_j}^{(c)}$ is abbreviated as $m_{t_j}$.

Furthermore, the parameters for the change pattern $c$ are defined as follows.
\begin{definition}
  The parameter $\bm{\Theta}^c$ for the change pattern $c$ is defined as follows.
  \begin{align} \label{parameter_def2}
    &\bm{\Theta}^c \defeq
    \left\{ \btheta^{m_t^{(c)}}\,\middle|\,t\in\mathcal{T}_c \right\} =
    \left\{ {\btheta^{m_{t_0}^{(c)}}},
    {\btheta^{m_{t_1}^{(c)}}}, \ldots,
    \btheta^{m_{t_{\abs{\mathcal{T}_c}-1}}^{(c)}} \right\}.
  \end{align}
\end{definition}

From Definition \ref{parameter_def},
using the state $S_m(x^{t-1})$ corresponding to the source sequence (past context) $x^{t-1}$ in a certain context tree model $m$, the probability of occurrence at time $t\in\left[t_j,t_{j+1}\right)$ is expressed as follows:
\begin{align} \nonumber
  p\left( x_t \middle|
  x^{t-1},\btheta^{m_{t_j}^{(c)}},m_{t_j}^{(c)},c \right) = &
  \,p\left( x_t \middle|
  x_{t_j}^{t-1},\btheta^{m_{t_j}},m_{t_j} \right) \\ \label{parameter2} = &
  \,\theta_{ x_t | S_{m_{t_j}} ( x_{t_j}^{t-1} ) }.
\end{align}
Therefore, the probability distribution of $X^N=X_1\cdots X_N$ in the change pattern $c\in\mathcal{C}$ is expressed by the following equation.
\begin{align} \nonumber
  p\left( x^N\middle|\bm{\Theta}^c,c\right) =&
  \prod_{j=0}^{\abs{\mathcal{T}_c}-1}
  p\left( x_{t_j}^{t_{j+1}-1} \middle| \btheta^{m_{t_j}}, m_{t_j} \right) \\ \nonumber =&
  \prod_{j=0}^{\abs{\mathcal{T}_c}-1} \prod_{t=t_j}^{t_{j+1}-1}
  p\left( x_t \middle| x_{t_j}^{t-1},\btheta^{m_{t_j}},m_{t_j} \right) \\ =&
  \prod_{j=0}^{\abs{\mathcal{T}_c}-1} \prod_{t=t_j}^{t_{j+1}-1}
  {\theta_{ x_t | S_{m_{t_j}} ( x_{t_j}^{t-1} ) }}.
\end{align}
Regarding a change pattern $c$, a context tree model $m_{t_j}^{(c)}$, and the parameter $\btheta^{m_{t_j}^{(c)}}$ for $m_{t_j}^{(c)}$,
we assume prior probability distributions $\pi(c)$, $P(m_{t_j}^{(c)}|c)$,
and $w(\btheta^{m_{t_j}^{(c)}}|m_{t_j}^{(c)},c)$ respectively.

\section{Bayes Code for the Proposed Source} \label{ベイズ符号化}

In this section, we present the coding probability of the Bayes code for the proposed source.
\begin{theorem} \label{bayes_theorem}
  The coding probability of the sequential Bayes code for the proposed source is
  \begin{align} \nonumber
    & \mathrm{AP}^\ast(x_t|x^{t-1}) =
    \sum_{c\in\mathcal{C}} \pi(c|x^{t-1}) \left[
    \sum_{m^{(c)}\in\mathcal{M}} P(m^{(c)}|x^{t-1}, c) \right. \\ \label{bayes_optimal} & ~
    \!\int\! p\!\left( x_t \middle| x^{t-1}\!,\!\btheta^{m^{(c)}}\!,\!m^{(c)}\!,\!c \right)
    \!\left.\! w\!\left( \btheta^{m^{(c)}} \middle| x^{t-1}\!,\! m^{(c)}\!,\! c \right)
    \!d\btheta^{m^{(c)}} \!\right],
  \end{align}
  where $\pi(c|x^{t-1})$ is the posterior probability distribution of the change pattern $c$,
  $P(m^{(c)}|x^{t-1}, c)$ is that of the context tree model $m^{(c)}$,
  $\mathcal{M}$ is the set of all context tree models\footnote{To be more precise, the size of $\mathcal{M}$ (the total number of context tree models) depends on the maximum depth of the context tree $d$.
  However, in this paper, $d$ is fixed.} (see Example \ref{example1}), and
  $w(\btheta^{m^{(c)}}|x^{t-1},m^{(c)},c)$ is the posterior probability distribution of the parameter $\btheta^{m^{(c)}}$.
\end{theorem}
\begin{proof}
  The proof outline of Theorem \ref{bayes_theorem} is the same as that of Theorem 2 in \cite{matsu_bayes}.
  Now, let $\mathrm{AP}_p(x_t|x^{t-1})$ be an arbitrary sequential coding probability.
  Taking the logarithmic loss of the coding probability, the loss function is as follows:
  \begin{align} \nonumber
    & V\left( \mathrm{AP}_p, x^N, \btheta^{m^{(c)}}, m^{(c)}, c \right) \\ & =
    \log p\left( x^N \middle| \btheta^{m^{(c)}}, m^{(c)}, c \right) -
    \log \prod_{t=1}^N \mathrm{AP}_p(x_t|x^{t-1}).
  \end{align}
  We take the expectation with respect to the probability distribution of the source sequences, and obtain the risk function as follows:
  \begin{align} \nonumber
    & R\left( \mathrm{AP}_p, \btheta^{m^{(c)}}, m^{(c)}, c \right) \\ & =
    \sum_{x^N} p\left( x^N \middle| \btheta^{m^{(c)}}, m^{(c)}, c \right)
    \log \frac{ p\left( x^N \middle| \btheta^{m^{(c)}}, m^{(c)}, c \right) }
    { \prod_{t=1}^N \mathrm{AP}_p(x_t|x^{t-1}) }.
  \end{align}
  The Bayes risk is then obtained by taking the expectation with respect to the probability distribution of each parameter.
  \begin{align} \nonumber
    & \mathrm{BR}(\mathrm{AP}_p) =
    \sum_{c\in\mathcal{C}} \pi(c) \left[
    \sum_{m^{(c)}\in\mathcal{M}} P(m^{(c)}|c) \right. \\ & \hspace{3mm} \left.
    \int R\left( \mathrm{AP}_p, \btheta^{m^{(c)}}, m^{(c)}, c \right)
    w(\btheta^{m^{(c)}} | m^{(c)}, c) d\btheta^{m^{(c)}} \right].
  \end{align}
  On the other hand, we have
  \begin{align} \nonumber
    & p(x^N) \\ & = \nonumber
    \sum_{c\in\mathcal{C}} \sum_{m^{(c)}\in\mathcal{M}} \int
    p(\btheta^{m^{(c)}}, m^{(c)}, c)
    p(x^N | \btheta^{m^{(c)}}, m^{(c)}, c) d\btheta^{m^{(c)}} \\ & =
    \prod_{t=1}^N \frac{ \sum_{c} \sum_{m^{(c)}} \! \int \!
    p(\btheta^{m^{(c)}} \! , \! m^{(c)} \! , \! c)
    p(x^t | \btheta^{m^{(c)}} \! , \! m^{(c)} \! , \! c)
    d\btheta^{m^{(c)}} }
    { \sum_{c} \sum_{m^{(c)}} \! \int \!
    p(\btheta^{m^{(c)}} \! , \! m^{(c)} \! , \! c)
    p(x^{t-1} \! | \btheta^{m^{(c)}} \! , \! m^{(c)} \! , \! c)
    d\btheta^{m^{(c)}} }.
  \end{align}
  Hence, $\mathrm{AP}^\ast(x_t|x^{t-1})$ given as follows minimizes the Bayes risk.
  \begin{align*}
    & \mathrm{AP}^\ast(x_t|x^{t-1}) \\ & =
    \frac{ \sum_{c} \sum_{m^{(c)}} \int
    p(\btheta^{m^{(c)}} \! , \! m^{(c)} \! , \! c)
    p(x^t \! | \btheta^{m^{(c)}} \! , \! m^{(c)} \! , \! c)
    d\btheta^{m^{(c)}} }
    { \sum_{c} \sum_{m^{(c)}} \int
    p(\btheta^{m^{(c)}} \! , \! m^{(c)} \! , \! c)
    p(x^{t-1} \! | \btheta^{m^{(c)}} \! , \! m^{(c)} \! , \! c)
    d\btheta^{m^{(c)}} } \\ & = \nonumber
    \sum_{c\in\mathcal{C}} \pi(c|x^{t-1}) \left[
    \sum_{m^{(c)}\in\mathcal{M}} P(m^{(c)}|x^{t-1}, c) \right.
  \end{align*}
  \vspace{-5mm}
  \begin{align}
    ~ \int p\!\left( x_t \middle| x^{t-1} \! , \! \btheta^{m^{(c)}} \! , \! m^{(c)} \! , \! c \right)
    \left. \!
    w\!\left( \btheta^{m^{(c)}} \! \middle| x^{t-1} \! , \! m^{(c)} \! , \! c \right)
    \! d\btheta^{m^{(c)}} \! \right].
  \end{align}
\end{proof}
\begin{example} \label{example1}
  When considering the case where $\mathcal{X}=\{ 0,1\}$ and the depth of a context tree model is at most two, there are five context tree models as shown in Figure \ref{models} and $\mathcal{M} = \{m_1,m_2,m_3,m_4,m_5\}$.
\end{example}
\begin{figure}[ht]
  \centering
  \includegraphics[width=0.8\columnwidth, bb=0 0 519 415]{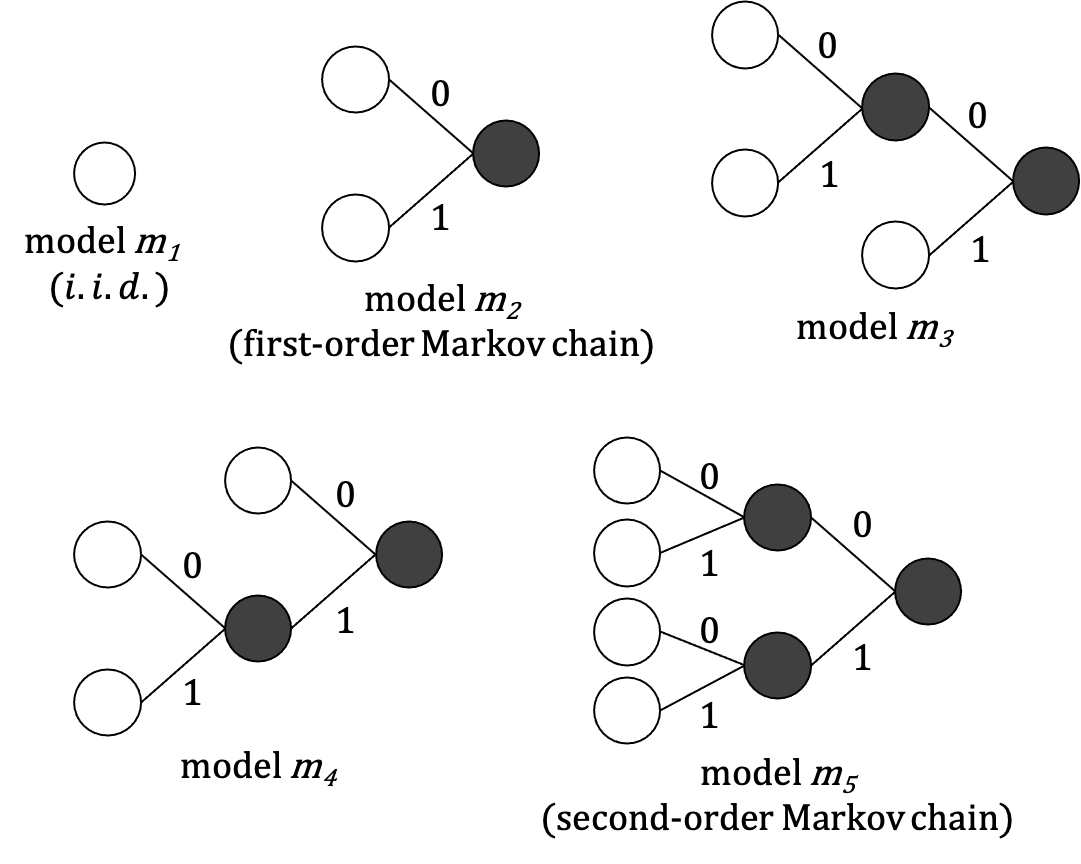}
  \caption{Overall view of the context tree models of $0$-$1$ sequence in the maximum depth $d=2$}
  \label{models}
\end{figure}

\section{Efficient algorithm for calculating coding probability of the Bayes code} \label{algorithm}

The coding probability of the Bayes code shown in Theorem \ref{bayes_theorem} is given by weighting the Bayes optimal coding probability of each change pattern $c$ by the posterior probability distribution of the change pattern $\pi(c|x^{t-1})$.
That is, \eqref{bayes_optimal} is expressed as follows:
\begin{align} \label{bayes_optimal2}
  \mathrm{AP}^\ast(x_t|x^{t-1}) =
  \sum_{c\in\mathcal{C}} \pi(c|x^{t-1}) \mathrm{AP}_c(x_t|x^{t-1}, c),
\end{align}
where
\begin{align} \nonumber
  & \mathrm{AP}_c(x_t|x^{t-1}, c) =
  \sum_{m^{(c)}\in\mathcal{M}} P\left(m^{(c)}\middle|x^{t-1}, c\right) \\ & ~
  \int p\!\left( x_t \middle| x^{t-1} \! , \! \btheta^{m^{(c)}} \! , \! m^{(c)} \! , \! c \right)
  w\!\left( \btheta^{m^{(c)}} \! \middle| x^{t-1} \! , \! m^{(c)} \! , \! c \right)
  \! d\btheta^{m^{(c)}}.
\end{align}

For this $\mathrm{AP}_c(x_t|x^{t-1}, c)$, Matsushima and Hirasawa\cite{matsu1_ctm} have already shown an algorithm that can calculate it analytically while reducing the amount of computation.
This algorithm is explained in the next subsection.

\subsection{Efficient Bayes Coding Algorithm for Fixed Change Pattern}

First, the prior probability distribution for the context tree model is assumed to be as follows.
\begin{assumption} \label{g_s prior}
  For the set of leaf nodes $L_m$ in each context tree model $m\in\mathcal{M}$, let $I_m$ be the set of internal nodes.
  Assume that each node $s$ has a hyper-parameter $g_s\in[0,1]$ and that the prior distribution of each model $m$ is
  \begin{align} \label{note2}
    P(m) = \prod_{\bar{s}\in I_m} g_{\bar{s}}
    \prod_{s\in L_m} (1-g_s),
  \end{align}
  where $g_s=0$ for the leaf node $s$ at the maximum depth of a context tree model. An example is shown in Figure~\ref{hyper_parameter}.
\end{assumption}
\begin{remark}
  $P(m)$ is a probability distribution, \ie\ \eqref{note2} satisfies
  \begin{align}
    \sum_{m\in\mathcal{M}} P(m) = 1.
  \end{align}
  The proof of this fact is given by Nakahara and Matsushima\cite{nakahara}.
\end{remark}
For example, the prior probability for Figure \ref{hyper_parameter} (corresponding to model $m_4$ in Figure \ref{models}) is
\begin{align} \nonumber
  P(m_4) = & g_{s_\lambda}(1-g_{s_0})g_{s_1}
  (1-\underbrace{g_{s_{01}}}_{0})
  (1-\underbrace{g_{s_{11}}}_{0}) \\ = &
  g_{s_\lambda}(1-g_{s_0})g_{s_1},
\end{align}
where $s_\lambda$ represents the root node.
\begin{figure}[ht]
  \centering
  \includegraphics[width=0.7\columnwidth, bb=0 0 430 311]{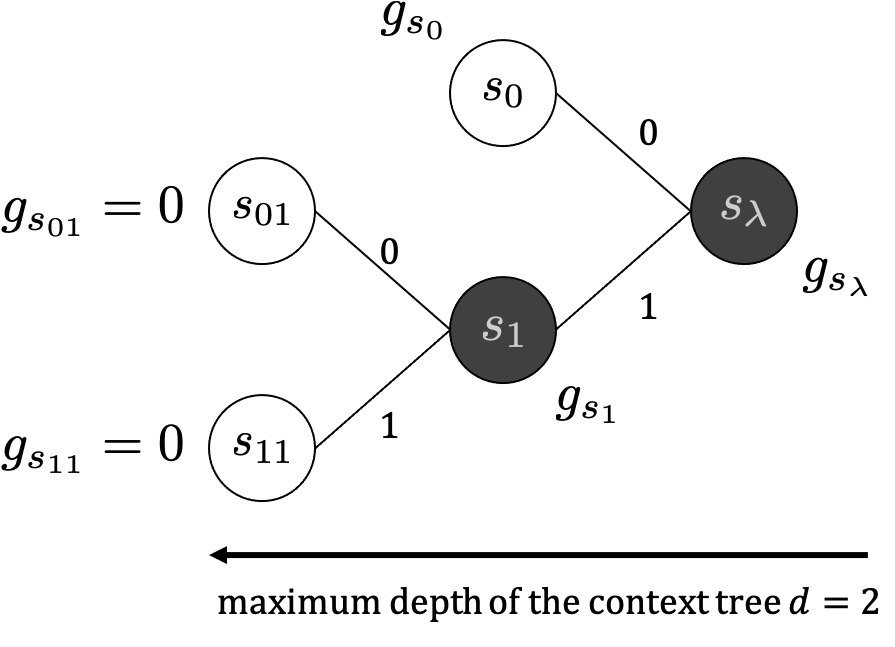}
  \caption{Example of a context tree model. The $s_\lambda$ represents the root node, and the internal nodes are $s_\lambda$ and $s_1$. The leaf nodes are $s_0,s_{01},s_{11}$, but the hyper-parameters of $s_{01},s_{11}$ are $0$ because they exist at the maximum depth.}
  \label{hyper_parameter}
\end{figure}

Second, we assume the prior probability distribution of the parameter $\btheta^m$ for each context tree model $m\in\mathcal{M}$ as follows.
\begin{assumption} \label{parameter_prior}
  For each leaf node $s\in L_m$, we assume the prior probability distribution of its parameter $w(\btheta_s)$ is a Dirichlet distribution
  \begin{align}
    w(\btheta_s) = \frac{ \Gamma\left(
    \sum_{i=0}^{\abs{\mathcal{X}}-1} \beta(i|s)\right) }
    { \prod_{i=0}^{\abs{\mathcal{X}}-1} \Gamma\left(\beta(i|s)\right) }
    \prod_{i=0}^{\abs{\mathcal{X}}-1} \theta_{i|s}^{\beta(i|s)-1},
  \end{align}
  where $\Gamma(\cdot)$ is a Gamma function, $\mathcal{X}$ denotes a source alphabet, and $\beta(i|s)$ denotes the parameter of the Dirichlet distribution.
  In addition, the prior probability distribution of $\btheta^m$ is assume to be the product of all $w(\btheta_s)$. That is,
  \begin{align}
    w\left( \btheta^m \middle| m \right) = \prod_{s\in L_m}
    w(\btheta_s).
  \end{align}
\end{assumption}

Then, Matsushima and Hirasawa\cite{matsu1_ctm} recursively compute the Bayes coding probability for context tree sources as follows.

First, let $\mathcal{L}$ denote the set of all leaf nodes in the \textit{superposed context tree}.
The term ``superposed context tree'' refers to the tree structure that represents the superposition of the entire possible context tree models.
For example, a superposed context tree for the entire context tree model shown in Figure \ref{models} is as shown in Figure \ref{superposed}.
\begin{figure}[ht]
  \centering
  \includegraphics[width=0.5\columnwidth, bb=0 0 315 268]{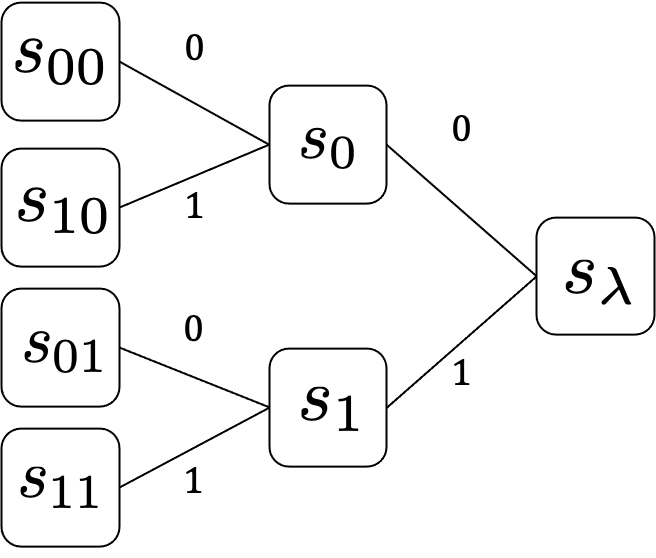}
  \caption{A superposed context tree for the entire context tree model in Figure \ref{models}}
  \label{superposed}
\end{figure}

Now, we denote $\tau_t$ as the last change point of context tree model when the time is $t$.
If $\tau_t$ is the $j$-th changing point in the sequence $x^{t-1}$, $x_t$ is generated according to the context tree model $m_{t_j}$.
In other words, since the context tree models before $\tau_t$ is irrelevant to $x_t$.
Next, a recursive function that calculates coding probability is defined as follows.
\begin{definition}
  \begin{align} \nonumber
    & \tilde{q}_s(x_t|x_{\tau_t}\cdots x_{t-1}) \\ \label{recur} & \defeq
    \begin{cases}
      q_s(x_t|x_{\tau_t}\cdots x_{t-1}) & \mbox{if}~~s\in \mathcal{L}, \\
      ( 1-g_{s|x_{\tau_t}\cdots x_{t-1}} ) q_s(x_t|x_{\tau_t}\cdots x_{t-1}) \\ +~
      g_{s|x_{\tau_t}\cdots x_{t-1}} \tilde{q}_{s_{\mbox{\tiny{child}}}}(x_t|x_{\tau_t}\cdots x_{t-1}) &
      \mbox{otherwise},
    \end{cases}
  \end{align}
  where
  \begin{align}
    q_s(x_t|x_{\tau_t}\cdots x_{t-1}) \defeq \frac{ \beta(x_t|s) + N(x_t|x_{\tau_t}\cdots x_{t-1},s) }
    { \sum_{i=0}^{\abs{\mathcal{X}}-1}
    \left\{ \beta(i|s)+N(i|x_{\tau_t}\cdots x_{t-1},s) \right\} }.
  \end{align}
  Note that $N(\mathrm{a}|x_{\tau_t}\cdots x_{t-1},s)$ denotes the number of occurrences of the symbol $\mathrm{a}\in\mathcal{X}$ under the state $s$ in the subsequence $x_{\tau_t}\cdots x_{t-1}$,
  and $s_{\mbox{\tiny{child}}}$ is the child node of $s$ in the superposed context tree with the context $x_{\tau_t}\cdots x_{t-1}$.
  The posterior hyper-parameter $g_{s|x_{\tau_t}\cdots x_{t-1}}$ is calculated as follows:
  \begin{align} \label{child}
    g_{s|x_{\tau_t}\cdots x_{t}} \defeq
    \begin{cases}
      g_s & \mbox{if $t=0$,} \\
      \frac{ g_{s|x_{\tau_t}\cdots x_{t-1}}
      \tilde{q}_{s_{\mbox{\tiny{child}}}}(x_t|x_{\tau_t}\cdots x_{t-1}) }
      { \tilde{q}_s(x_t|x_{\tau_t}\cdots x_{t-1}) } & \mbox{otherwise.}
    \end{cases}
  \end{align}
\end{definition}

In this case, the Bayes coding probability $\mathrm{AP}_c(x_t|x^{t-1}, c)$ for the context tree source can be calculated as follows.
\begin{theorem}[Matsushima and Hirasawa\cite{matsu1_ctm}]
  \begin{align}
    \mathrm{AP}_c(x_t|x^{t-1}, c) = \tilde{q}_{s_\lambda}(x_t|x_{\tau_t}\cdots x_{t-1}).
  \end{align}
\end{theorem}
%
%
\subsection{Efficient Bayes Coding Algorithm for the Proposed Source}

In the previous subsection, we described that the algorithm by Matsushima and Hirasawa\cite{matsu1_ctm} reduces the computational complexity of calculating $\mathrm{AP}_c(x_t|x^{t-1}, c)$ for each change pattern $c$ from $\mathcal{O}(2^{\abs{\mathcal{X}}^{d-1}})$ to $\mathcal{O}(d)$.
However, \eqref{bayes_optimal2} for the proposed source requires a computational effort of $\mathcal{O}(d\cdot2^N)$ because the total number of change patterns is $\abs{\mathcal{C}}=2^N$ for the length $N$ of the source sequence.

In this subsection, we propose an algorithm to reduce the computational complexity to $\mathcal{O}(d\cdot N^2)$ using a class of prior probability distributions of change patterns by Suko {\etal}\cite{suko_ipid}.
In their proposal of efficient Bayes coding for {\ipid} sources, they assumed a class that follows a Bernoulli distribution for the pattern of parameter changes.
In this paper, we assume a similar class for the change pattern $c$.
\begin{definition} \label{bernoulli}
  For the change pattern $c=( w_1^{(c)},\ldots,w_t^{(c)},\ldots,w_N^{(c)} )$, we assume that each $w_t^{(c)}$ (except $t=1$) independently follows Bernoulli distribution $\mathrm{Ber}(\alpha)$ for each $w_t^{(c)}$ (except $t=1$). That is,
  \begin{align}
    \pi(c) \defeq \alpha^{\abs{\mathcal{T}_c}-1}
    { (1-\alpha)^{N-\abs{\mathcal{T}_c}} }.
  \end{align}
\end{definition}
In addition, Suko {\etal}\cite{suko_ipid} introduced a prior probability distribution of $\tau_t$, and proposed an efficient Bayes code for {\ipid} sources.
In the same way, we define the prior probability distribution of the last change point $v(\tau_t)$ as follows:
\begin{align}
  v(\tau_t) \defeq \sum_{c:\tau_t\in\mathcal{T}_c} \pi(c), ~
  \mbox{where} ~ \tau_t = 1,2,\ldots,t.
\end{align}

Finally, the efficient Bayes coding algorithm for the proposed model in this paper is shown below.
\begin{screen}
  \begin{enumerate}[Step i.]
    \item Load $x_t$. \label{step1}
    \item The coding probability is calculated as follows:
          \begin{align*} \label{calculation} \hspace{-1.2cm}
            \tilde{p}(x_t|x^{t-1}) = \sum_{\tau_t=1}^t
            \tilde{q}_{s_\lambda}(x_t|x_{\tau_t}\cdots x_{t-1}) v(\tau_t|x^{t-1}).
          \end{align*}
    \item $v(\tau_{t+1}|x^t)$ is calculated as follows:
          \begin{itemize}
            \setlength{\leftskip}{-1.2cm}
            \item If $\tau_{t+1}=1,2,\ldots,t$,
                  \begin{align*} \hspace{-1.2cm}
                    v(\tau_{t+1}|x^t) = (1-\alpha)
                    \frac{\tilde{q}_{s_\lambda}(x_t|x_{\tau_t}\cdots x_{t-1})v(\tau_t|x^{t-1})}
                    {\tilde{p}(x_t|x^{t-1})}.
                  \end{align*}
            \item If $\tau_{t+1}=t+1$, $v(\tau_{t+1}|x^t)=\alpha$.
          \end{itemize}
    \item Back to Step \ref{step1}.
  \end{enumerate}
\end{screen}

We show that the above algorithm correctly computes the Bayes coding probability for the proposed source.
\begin{proof}
  (Outline only.)
  \begin{align} \nonumber
    \mathrm{AP}^\ast(x_t|x^{t-1}) = &
    \sum_{c\in\mathcal{C}} \pi(c|x^{t-1}) \mathrm{AP}_c(x_t|x^{t-1}, c) \\ \nonumber = &
    \sum_{c\in\mathcal{C}} \pi(c|x^{t-1}) \tilde{q}_{s_\lambda}(x_t|\tau_t, x^{t-1}) \\ \nonumber = &
    \sum_{\tau_t=1}^t \left\{ \tilde{q}_{s_\lambda}(x_t|\tau_t, x^{t-1})
    \sum_{c:\tau_t\in\mathcal{T}_c} \pi(c|x^{t-1}) \right\} \\ \nonumber = &
    \sum_{\tau_t=1}^t \tilde{q}_{s_\lambda}(x_t|\tau_t, x^{t-1}) v(\tau_t|x^{t-1}) \\ = &
    \ \tilde{p}(x_t|x^{t-1}).
  \end{align}
\end{proof}

\section{Experiment} \label{experiment}

We performed an experiment of running the proposed algorithm on an artificially generated source sequence.
The purpose of this experiment is to check the compression performance of the proposed algorithm for several settings of hyper-parameter $\alpha$ (see Definition \ref{bernoulli}).

First, we describe the source sequence we used in the experiment.
The length of the sequence is $300$, and as shown in Figure \ref{source}, each of the $100$ consecutive symbols is generated from a different context tree model.
\begin{figure}[ht]
  \centering
  \includegraphics[width=1.0\columnwidth, bb=0 0 1017 279]{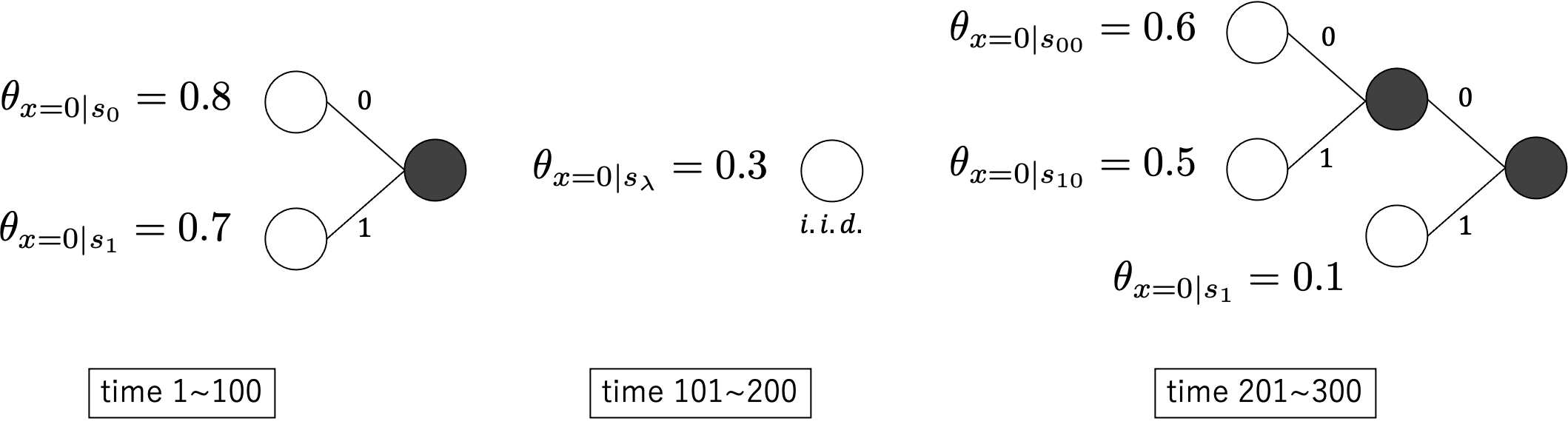}
  \caption{Context tree models that generate the source sequence we used in the experiment. Note that $\theta_{x=1|s_\lambda}$ is given by
  $1-\theta_{x=0|s_\lambda}$.}
  \label{source}
\end{figure}
Since the context tree model of the source sequence changes every $100$ symbols,
the appropriate hyper-parameter (the parameter of the Bernoulli distribution in Definition \ref{bernoulli}) $\alpha$ is $0.01$.
Therefore, in order to confirm the compression performance for the values of $\alpha$, we conducted the following Experiment A.

\hspace{-5mm} \underline{Experiment A} \vspace{1mm}

In Experiment A, we observe the redundancy
\begin{align}
  \log_2 \frac{1}{\tilde{p}(x_t|x^{t-1})} - H\left(X_{\tau_t}^{t}\right)
\end{align}
of the proposed algorithm, where $H(\cdot)$ denotes entropy rate.
We set the hyper-parameter $\alpha$ with three values: $\alpha=0.1$, $0.01$, and $0.001$.
We ran the algorithm $10000$ times with each value of $\alpha$.
In the algorithm, we set $g_s=0.5$ (see Assumption \ref{g_s prior}) and $\beta(0|s)=\beta(1|s)=0.5$ (see Assumption \ref{parameter_prior}).
The average redundancy of the $10000$ times is taken and the result is shown in Figure \ref{redundancy}.
\begin{figure}[ht]
  \centering
  \includegraphics[width=0.9\columnwidth, bb=0 0 499 416]{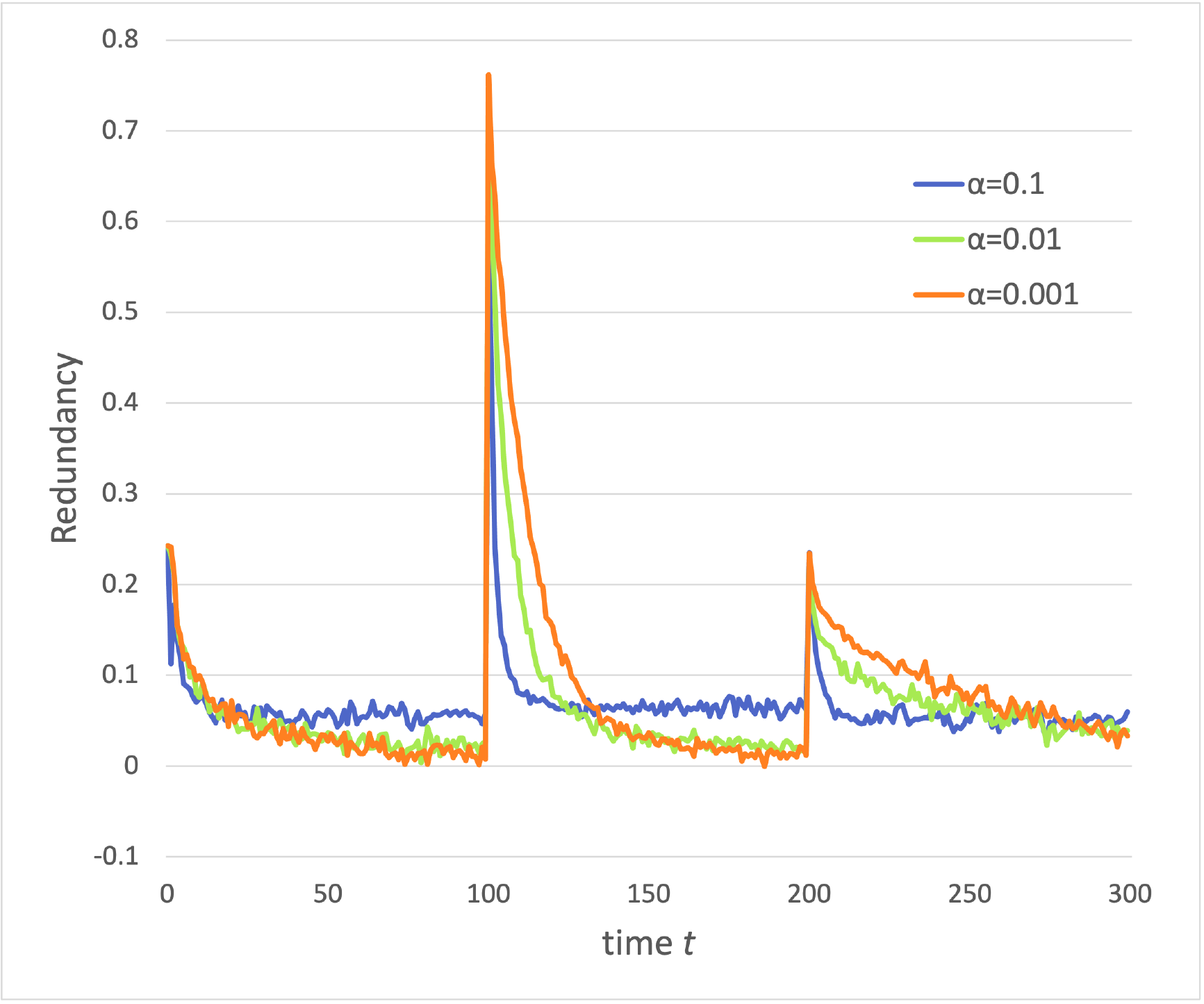}
  \caption{Average redundancy in the $10000$ times trial}
  \label{redundancy}
\end{figure}

The redundancy jumps up instantly at each changing point ($t=101, 201$), but decreases gradually in any values of $\alpha$.
\begin{description}
  \item[($\alpha=0.1$ vs. $\alpha=0.01$)]\mbox{}\\
  When $\alpha=0.1$ the redundancy decreases more rapidly.
  However, it converges to a smaller value when $\alpha=0.01$ at the end of each interval.
  \item[($\alpha=0.01$ vs. $\alpha=0.001$)]\mbox{}\\
  When $\alpha=0.01$, the redundancy decreases more rapidly.
  Moreover, there is no slight difference of the convergence value at the end of each interval.
\end{description}
Therefore, it can be concluded that $\alpha=0.01$ seems the best for compression performance.

In the Bayes code, we can observe the posterior probability of parameters.
For example, the posterior probability of $\tau_t$ indicates the characteristics of context tree model changes.
Next, we conducted the following Experiment B.

\hspace{-5mm} \underline{Experiment B} \vspace{1mm}

The $\tau_t$ with the largest posterior probability $v(\tau_t|x^{t-1})$ is regarded as the estimate of the last change point from $t$, and we observe the transition of the estimate.
As same as Experiment A, we set the three hyper-parameter $\alpha=0.1$, $0.01$, and $0.001$. We ran the algorithm $10000$ times with each hyper-parameter.
The average of the estimate of the $10000$ times is taken and the result is shown in Figure \ref{estimation}.
\begin{figure}[ht]
  \centering
  \includegraphics[width=0.9\columnwidth, bb=0 0 499 416]{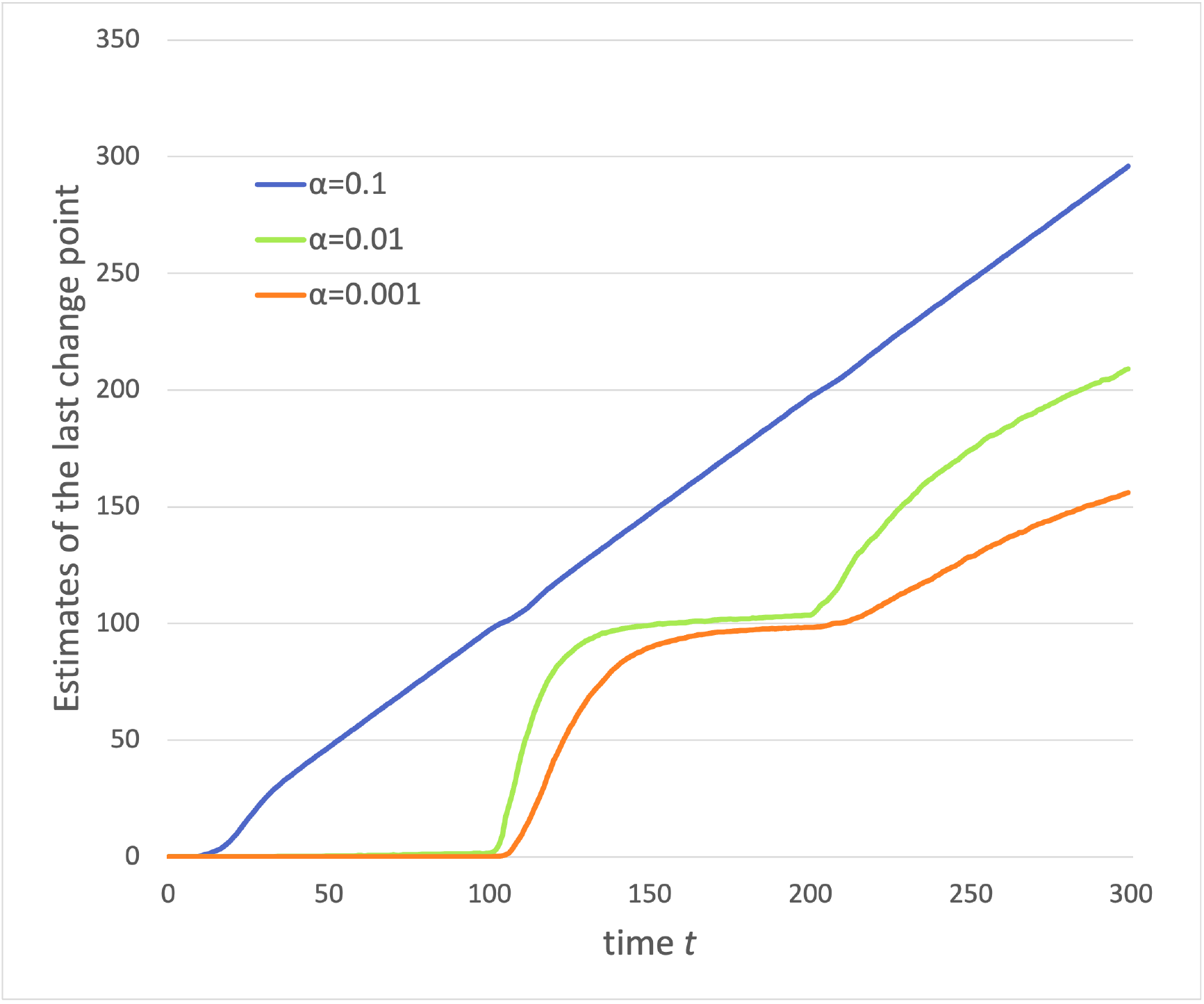}
  \caption{Average of the estimates of the last change point in the $10000$ times trial}
  \label{estimation}
\end{figure}

When $\alpha=0.01$ and $0.001$, it seems that the transition of $\argmax_{\tau_t} v(\tau_t|x^{t-1})$ corresponds to the change of context tree model.
When $\alpha=0.1$, however, $\argmax_{\tau_t} v(\tau_t|x^{t-1}) \simeq t$, so it does not correspond to the change of context tree model.
\section*{Acknowledgement}
This work was supported in part by JSPS KAKENHI Grant Numbers JP17K06446, JP19K04914, and JP19K14989.

\end{document}